\documentclass[11pt]{article}
\usepackage[margin=1in]{geometry} 
\usepackage[colorlinks,
             linkcolor=black!50!red,
             citecolor=blue,
             pdftitle={Defects in Spin Chains via Cluster Categories},
             pdfauthor={Ammar Husain},
             pdfsubject={Spin Chains},
             pdfkeywords={spin chains,cluster algebra,Picard,K-theory,bimodules},
             pdfproducer={pdfLaTeX},
             pdfpagemode=None,
             bookmarksopen=true
             bookmarksnumbered=true]{hyperref}

\usepackage{amsmath}
\usepackage{amssymb}
\usepackage{amsfonts,bbm,mathrsfs,stmaryrd}
\usepackage{url}

\usepackage[amsmath,thmmarks]{ntheorem}
\usepackage{hyperref}
\usepackage{cleveref}

\theoremstyle{change}

\newtheorem{definition}[equation]{Definition}

\newtheorem{theorem}[equation]{Theorem}

\newtheorem{lemma}[equation]{Lemma}
\newtheorem{cor}[equation]{Corollary}
\newtheorem{conj}[equation]{Conjecture}
\newtheorem{conjecture}[equation]{Conjecture}

\newtheorem{example}[equation]{Example}

\theorembodyfont{\upshape}
\theoremsymbol{\ensuremath{\Diamond}}

\newtheorem{remark}[equation]{Remark}

\theoremstyle{nonumberplain}

\theoremsymbol{\ensuremath{\Box}}
\newtheorem{proof}{Proof}

\qedsymbol{\ensuremath{\Box}}

\creflabelformat{equation}{#2(#1)#3} 

\crefname{equation}{equation}{equations}
\crefname{eg}{example}{examples}
\crefname{defn}{definition}{definitions}
\crefname{prop}{proposition}{propositions}
\crefname{thm}{Theorem}{Theorems}
\crefname{lemma}{lemma}{lemmas}
\crefname{cor}{corollary}{corollaries}
\crefname{remark}{remark}{remarks}
\crefname{section}{Section}{Sections}
\crefname{subsection}{Section}{Sections}

\crefformat{equation}{#2equation~(#1)#3} 
\crefformat{eg}{#2example~#1#3} 
\crefformat{defn}{#2definition~#1#3} 
\crefformat{prop}{#2proposition~#1#3} 
\crefformat{thm}{#2Theorem~#1#3} 
\crefformat{lemma}{#2lemma~#1#3} 
\crefformat{cor}{#2corollary~#1#3} 
\crefformat{remark}{#2remark~#1#3} 
\crefformat{section}{#2Section~#1#3} 
\crefformat{subsection}{#2Section~#1#3} 

\Crefformat{equation}{#2Equation~(#1)#3} 
\Crefformat{eg}{#2Example~#1#3} 
\Crefformat{defn}{#2Definition~#1#3} 
\Crefformat{prop}{#2Proposition~#1#3} 
\Crefformat{thm}{#2Theorem~#1#3} 
\Crefformat{lemma}{#2Lemma~#1#3} 
\Crefformat{cor}{#2Corollary~#1#3} 
\Crefformat{remark}{#2Remark~#1#3} 
\Crefformat{section}{#2Section~#1#3} 
\Crefformat{subsection}{#2Section~#1#3}

\numberwithin{equation}{section}

\usepackage{tikz}
\tikzset{dot/.style={circle,draw,fill,inner sep=1pt}}
\usepackage{braids}
\usepackage{tqft}
\usetikzlibrary{tqft}
\usetikzlibrary{cd}
\usetikzlibrary{arrows}
\usepackage[]{youngtab}
\usepackage{young}

\newcommand\setof[1]{\{ #1 \}}

\usepackage{scalerel,stackengine}
\stackMath
\newcommand\reallywidehat[1]{%
\savestack{\tmpbox}{\stretchto{%
  \scaleto{%
    \scalerel*[\widthof{\ensuremath{#1}}]{\kern-.6pt\bigwedge\kern-.6pt}%
    {\rule[-\textheight/2]{1ex}{\textheight}}
  }{\textheight}%
}{0.5ex}}%
\stackon[1pt]{#1}{\tmpbox}%
}
\usepackage{slashed}

\title{Defects in Spin Chains via Cluster Categories}
\author{Ammar Husain}

\begin{document}
\maketitle

\tikzset{->-/.style={decoration={
  markings,
  mark=at position #1 with {\arrow{>}}},postaction={decorate}}}
\tikzset{-<-/.style={decoration={
  markings,
  mark=at position #1 with {\arrow{<}}},postaction={decorate}}}

\abstract

We study Picard groups and $K_0$ groups for the cluster algebras that come from the cluster categories of \cite{HernandezLeclerc}. This is inspired by a study of boundaries and defects in the associated spin chains. We also do some discussion for their formal deformation quantizations.

\section{Introduction}

The algebraic Bethe ansatz is built from decorating the sites with representations of $U_q \hat{\mathfrak{g}}$. Boundaries for such systems are determined by giving a representation of a coideal subalgebra \cite{Letzter}. That gives a certain kind of module category. We may also talk about building bimodule categories in order to construct two-sided defects in the chain \cite{self}.

We would like to ask about invariants that help characterize (invertible) (bi-) module categories for these. However, that is complicated due to the nature of meromorphic tensor categories \cite{Soibelman} as well as potential failures in the monoidal nature of $K_0$. So we perform 2 simplifications. First we look at the subcategory defined in \cite{HernandezLeclerc} as cluster categories.

The next simplification is look at the ``reduced" theory. That is where there used to be the monoidal category $\mathcal{C}_l$ assigned to all the bulk regions, we instead assign $\mathcal{A} = K_0 ( \mathcal{C}_l ) \otimes \mathbb{Q}$. In this reduced theory, the codimension 1 walls are decorated by $\mathcal{A} - \mathcal{A}$ bimodules and codimension 2 phenomena have bimodule intertwiners \cite{DSPS}. Physical intuition says that there should be some procedure to go from a bimodule category which is being used as a defect wall in the original theory to a bimodule which is used as a defect in analogy with a 2-dimensional topological theory.\footnote{This is a mere analogy. It can not be pushed due to the failure of separable symmetric Frobenius algebra.} The naive guess of $K_0$ as the reduction fails unless the functor defining the module structures are bi-exact. That is an issue that has been sidestepped here by only asking about classifying (invertible) (bi)modules for these cluster algebras.

The use of the cluster category provides us with the tool of cluster algebras $\mathcal{A}$. These are significantly simpler algebras so the questions of their (bi)modules is significantly easier. That means we can focus on the problems of $Pic(\mathcal{A})$/$Aut(\mathcal{A})$ and $\mathcal{A}-(bi)mod$ needed for the invertible and not necessarily invertible cases respectively.

\section{Cluster Categories and Algebras}

\subsection{Cluster Categories}

\begin{definition}[Monoidal categorification]
For a cluster algebra $\mathcal{A}$, a monoidal categorification thereof is defined to be an abelian monoidal category such that it's Grothendieck ring is isomorphic to $\mathcal{A}$ and the cluster monomials are classes of real simple objects. Cluster variables are classes of real prime simple objects.
\end{definition}

Let us work with a $U_q \hat{\mathfrak{g}}$ for $\mathfrak{g}$ of ADE type of rank $n$  in the Drinfeld realization. Physically this means that we are in the Jordan-Wigner perspective rather than spins. This effects important properties because the different co-product means different entanglement structure. We are interested in representations of this as a category, but it is too large. That is the purpose of the following definition.

\begin{definition}[$\mathcal{C}_l$]
Color the vertices of the Dynkin diagram with $\xi_i = 0/1$. Now take the full subcategory of $Rep^{fd} U_q \hat{\mathfrak{g}}$ consisting of objects such that any simple composition factor and index $i \in I$, the roots of that Drinfeld polynomial are in the set $\setof{q^{-2k+\xi_i}}$ for $k \in \mathbb{Z}$. Call this $\mathcal{C}_\mathbb{Z}$.
Define $\mathcal{C}_l$ by only allowing $k \in [0,l]$ instead.\\
\end{definition}

\begin{lemma}[\cite{HernandezLeclerc}]
$K_0 (\mathcal{C}_\mathbb{Z})$ is generated by $[V_{i,q^{2k+\xi_i}}]$ and similarly for $K_0 (\mathcal{C}_l)$ but in the $\mathcal{C}_l$ case it is a polynomial ring in those $n(l+1)$ variables.\\
\end{lemma}

\subsection{Cluster Algebra}

In order to show that $\mathcal{C}_{\mathfrak{g},l}$ is a monoidal categorification, one must say what cluster algebra it is categorifying. That is given in via a quiver.

\begin{definition}[$Q_{\mathfrak{g},l}$]
Define a quiver for the Dynkin diagram for $\mathfrak{g}$ and $l$ by the following procedure.

Orient the diagram via a bipartitioning where all the vertices with $\xi_i = 0$ are colored black and those with $\xi_i = 1$ are colored white. Then orient the edges as going from black to white. Call this quiver $Q_{\mathfrak{g},0}$.

Now form a new quiver with vertex set $(i,k)$ for $i \in I$ and $1 \leq k \leq l+1$ with three types of arrows.

\begin{itemize}
\setlength\itemsep{-1em}
\item $(i,k) \to (j , k)$ whenever $i \to j$ is an arrow in $Q_{\mathfrak{g},0}$ and all $k$\\
\item $(j,k) \to (i,k+1)$ for every arrow $i \to j$ in $Q_{\mathfrak{g},0}$ and all $1 \leq k \leq l$\\
\item $(i,k+1) \to (i,k)$ for all $i \in I$ and all $1 \leq k \leq l$\\
\end{itemize}

\end{definition}

\begin{definition}[$\mathcal{A}_{\mathfrak{g},l}$]
Define the quiver cluster algebra with $x_{i,l+1}$ all frozen variables. This defines a cluster algebra $\mathcal{A}_{\mathfrak{g},l} \subset \mathbb{Q} ( x_{i,m} )$, inside a field extension of $n(l+1)$ transcendental variables.
\end{definition}

\begin{conjecture}[Leclerc]
There is a ring isomorphism $K_0 ( \mathcal{C}_{l}) \otimes \mathbb{Q} \simeq \mathcal{A}_{\mathfrak{g},l}$
\end{conjecture}

\begin{theorem}
The conjecture is true if $\mathfrak{g}=A_1$ and arbitrary $l$. In this case the quiver $Q_{A_1,l}$ is an $A_{l+1}$ with one frozen vertex at the end.\\
The conjecture is also true for $l=1$ in this case it is 2 copies connected up nontrivially.
\end{theorem}

\begin{figure}[htb!]
\centering
\begin{tikzpicture}[every node/.style={circle,draw},thick]
  \node[rectangle](N1) at (0,0){$W_{1,2}$};
  \node[rectangle](N2) at (2,0){$W_{2,2}$};
  \node[rectangle](N3) at (4,0){$W_{3,2}$};
  \node[rectangle](N4) at (6,0){$W_{4,2}$};
  \node[rectangle](M1) at (0,2){$V_{1,1}$};
  \node[rectangle](M2) at (2,2){$V_{2,1}$};
  \node[rectangle](M3) at (4,2){$V_{3,1}$};
  \node[rectangle](M4) at (6,2){$V_{4,1}$};
  \draw[<-,>=stealth',semithick](N1.east)--(N2.west);
\draw[->,>=stealth',semithick](N2.east)--(N3.west);
\draw[<-,>=stealth',semithick](N3.east)--(N4.west);
  \draw[<-,>=stealth',semithick](M1.east)--(M2.west);
\draw[->,>=stealth',semithick](M2.east)--(M3.west);
\draw[<-,>=stealth',semithick](M3.east)--(M4.west);
\draw[->,>=stealth',semithick](N1.north)--(M1.south);
\draw[->,>=stealth',semithick](N2.north)--(M2.south);
\draw[->,>=stealth',semithick](N3.north)--(M3.south);
\draw[->,>=stealth',semithick](N4.north)--(M4.south);
\draw[->,>=stealth',semithick](M1.south)--(N2.north);
\draw[->,>=stealth',semithick](M2.south)--(N3.north);
\draw[->,>=stealth',semithick](M3.south)--(N4.north);
\end{tikzpicture}
\caption{$\mathcal{A}_{A_4,1}$ with $W$ nodes being frozen.\label{exampleQuiver}}
\end{figure}
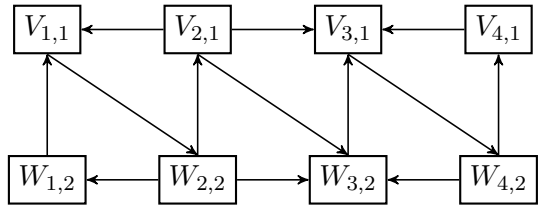

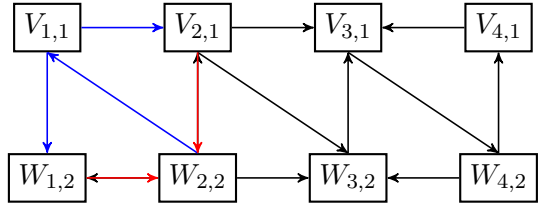
\begin{figure}[htb!]
\centering
\begin{tikzpicture}[every node/.style={circle,draw},thick]
  \node[rectangle](N1) at (0,0){$W_{1,2}$};
  \node[rectangle](N2) at (2,0){$W_{2,2}$};
  \node[rectangle](N3) at (4,0){$W_{3,2}$};
  \node[rectangle](N4) at (6,0){$W_{4,2}$};
  \node[rectangle](M1) at (0,2){$V_{1,1}$};
  \node[rectangle](M2) at (2,2){$V_{2,1}$};
  \node[rectangle](M3) at (4,2){$V_{3,1}$};
  \node[rectangle](M4) at (6,2){$V_{4,1}$};
  \draw[<-,>=stealth',semithick](N1.east)--(N2.west);
\draw[->,>=stealth',semithick](N2.east)--(N3.west);
\draw[<-,>=stealth',semithick](N3.east)--(N4.west);
  \draw[->,>=stealth',semithick,color=blue](M1.east)--(M2.west);
\draw[->,>=stealth',semithick](M2.east)--(M3.west);
\draw[<-,>=stealth',semithick](M3.east)--(M4.west);
\draw[<-,>=stealth',semithick,color=blue](N1.north)--(M1.south);
\draw[->,>=stealth',semithick](N2.north)--(M2.south);
\draw[->,>=stealth',semithick](N3.north)--(M3.south);
\draw[->,>=stealth',semithick](N4.north)--(M4.south);
\draw[<-,>=stealth',semithick,color=blue](M1.south)--(N2.north);
\draw[->,>=stealth',semithick](M2.south)--(N3.north);
\draw[->,>=stealth',semithick](M3.south)--(N4.north);
\draw[->,>=stealth',semithick,color=red](M2.south)--(N2.north);
\draw[->,>=stealth',semithick,color=red](N1.east)--(N2.west);
\end{tikzpicture}
\caption{Mutating at the node $V_{1,1}$. Red-inserted, Blue-flipped, need to remove 2-cycles.}
\end{figure}

\subsection{Automorphisms}

\begin{lemma}[2.21 of \cite{ChangZhu}\label{221}]
For a finite type cluster algebra $\mathcal{A}$ and $\mathcal{A}^{ex}$ the principal (unfrozen) part. Assuming $\mathcal{A}$ is gluing free, then the specialization map allows cluster automorphisms of $\mathcal{A}$ to give cluster automorphisms of $\mathcal{A}^{ex}$. That is $Aut(\mathcal{A}) \subset Aut(\mathcal{A}^{ex})$.
\end{lemma}

The cases of $Q_{\mathfrak{g},1}$ and $Q_{A_1,l}$ will have the property that their mutable parts are orientations of a simply laced Dynkin diagram so this lemma will apply for them.

\begin{lemma}
There is a homomorphism $Pic \mathcal{A} \to Pic \mathcal{A}^{ex}$
\end{lemma}

\begin{proof}
Specialization gives an algebra map $\mathcal{A} \to \mathcal{A}^{ex}$
\end{proof}

\begin{theorem}[\label{clustAut} Automorphisms of finite type\cite{ChangZhu}]
The cluster automorphisms of the principal parts for the $\mathcal{A}_{\mathfrak{g},1}$ are:

\begin{figure}[htb!]
\centering
\begin{tabular}{ l | r }
$\mathfrak{g}$ & Aut\\
\hline
$A_1$ & $\mathbb{Z}_2$\\
$A_{n > 1}$ &$D_{n+3}$\\
$D_4$&$D_{4} \times S_3$\\
$D_{n \geq 5}$&$ D_n \times \mathbb{Z}_2$\\
$E_6 $&$D_{14}$\\
$E_7 $&$ D_{10}$\\
$E_8 $&$D_{16}$\\
\end{tabular}
\end{figure}

\end{theorem}

This is much more specific then all automorphisms as rings but we still get a subgroup of all automorphisms. These types of automorphisms must take clusters to clusters and commute with mutations. By the lemma \cref{221}, we must look at a subgroup to account for the frozen variables.

\section{Bimodule Categories}

When looking at the reflection equation we produce coideal subalgebras whose representations produce module categories. Similarly for two-sided defects we desire to produce comodule algebras for two copies using a folding trick. In summary we seek to produce bimodule categories.

\begin{theorem}[\cite{DSPS}]
In the context of 3-dimensional topological field theories, this happens without the affinization. That is bimodule categories produce codimension 1 strata in the cobordism hypothesis with singularities. They may or may not be invertible.
\end{theorem}

\begin{lemma}
If $\mathcal{C}$ is the original category and $\mathcal{D}$ is the bimodule category. Let us also assume the $\mathcal{C} \times \mathcal{D} \to \mathcal{D}$ and vice versa are bi-exact. Then $K_0 (\mathcal{D})$ is both a left and right $K_0 ( \mathcal{C})$ bimodule. We may extend to $\mathbb{Q}$ and potentially produce a bimodule over $K_0 ( \mathcal{C}) \otimes \mathbb{Q}$ but now $\mathbb{Q}$ linear.

If $\mathcal{C}_1 \subset \mathcal{C}_2$ is a full subcategory and $\mathcal{D}$ is a bimodule category (satisfying the exactness assumptions) for $\mathcal{C}_2$ we may restrict $K_0 ( \mathcal{D}) \otimes \mathbb{Q}$ as a $K_0 (\mathcal{C}_2) \otimes \mathbb{Q} - K_0 (\mathcal{C}_2) \otimes \mathbb{Q}$ bimodule to $K_0 (\mathcal{C}_1) \otimes \mathbb{Q} - K_0 (\mathcal{C}_1) \otimes \mathbb{Q}$.
\end{lemma}

Here we let $\mathcal{C}_1$ be the cluster category and $\mathcal{C}_2$ be without constraints on the Drinfeld polynomials.

A major caution here is that we have not described any operation of fusing bimodule categories. Therefore any statements about the fusing of bimodules over cluster algebras (including the criterion of invertibility) does not necessarily lift to the category level. That is we are not talking about bimodule categories that describe invertible defects, but instead simply those that go to invertible objects in the simplification. This also means that the group operation on $Pic/K_0$ is only useful in the reduction.

\section{Picard Groups and $K_0$}

We may define $Pic_k (A)$ as invertible bimodules for a commutative ring $A$ where we only demand that $k$ be able to commute through the bimodule. This contrasts with the module definition ordinarily used. However, we may reduce the computation to two steps.

\begin{theorem}[Pic as bimodule \cite{Yekutieli}]
$Pic_k (A) \simeq Aut_k (A) \ltimes Pic_A (A)$ where $Pic_A (A)$ is equivalent to the usual $Pic_{com}$ with invertible modules rather than bimodules. $Aut_k (A)$ is the automorphisms as a $k$ algebra.
\end{theorem}

\subsection{Useful $Pic_{com}$ Groups}

Because we are considering cluster algebras, the $\mathbb{Q}$ algebras will be very special. They will be very similar to simple polynomial and Laurent polynomial rings.

\begin{theorem}[1.6 of \cite{Weibel} \label{polPic}]

\begin{eqnarray*}
Pic_{com} (\mathbb{Q}[t_1^\pm \cdots t_m^\pm , x_1 \cdots x_n]) = 0
\end{eqnarray*}

\end{theorem}

\begin{proof}
In fact true for any zero dimensional commutative ring so field is overkill. $\mathbb{Z}$ would not work because it has Krull dimension 1, but finite rings $\mathbb{Z}_n$ would.\\
\end{proof}

More generally we get

\begin{theorem}[\cite{Weibel}\label{Weibel}]
\begin{eqnarray*}
Pic_{com} (A[t_1^\pm \cdots t_m^\pm ]) \simeq Pic_{com} (A) \bigoplus \bigsqcup_{i=1}^m LPic_{com}(A) \bigoplus \bigsqcup_{k=1}^m \bigsqcup_{i=1}^{2^k {m}\choose{k}} N^k Pic_{com}(A)\\
\end{eqnarray*}

$LPic_{com}(A)$ for an anodal 1 dimensional domain is trivial. Anodal means that if $b \in \bar{A}$ the integral closure and $b^3 - b^2$ and $b^2 - b$ are in $A$, then $b \in A$ as well.

$NPic_{com}(A) \equiv Pic_{com}(A[t])/Pic_{com}(A)$. For example, $NPic_{com}(A) =0$ if and only if $A_{red}$ is seminormal. Normal rings like $U[x_1 \cdots x_n]$ for a UFD $U$ are a fortiori seminormal.

\end{theorem}

\begin{theorem}[2.2 of \cite{Coykendall} \label{ZpolPic}]
$Pic_{com} R \simeq Pic_{com} R[x_1 \cdots x_n]$ if and only if $R$ is seminormal. For example,
\begin{eqnarray*}
Pic_{com} \mathbb{Z} [x_1 \cdots x_n] &\simeq& 0\\
\end{eqnarray*}

\end{theorem}

\subsection{Automorphisms \label{polAut}}

Here we gather useful automorphism groups for the commutative algebras that may show up as the cluster algebras. Easy ones include $Aut \mathbb{Z}[x] = \mathbb{Z}_2 \ltimes \mathbb{Z}$. However we had at the very simplest a polynomial ring in $n (l+1)$ variables. Always $\geq 2$ variables. That means a complicated automorphism group. It is even complicated for $\mathcal{C}_{A_1 , 1}$ as described below.

\begin{definition}[Jonqui`ere group]

$J_n (R)$ is defined as those automorphisms of $R[x_1 \cdots x_n]$ of the form

\begin{eqnarray*}
x_1 &\to& F_1 (x_1) = \alpha_1 x_1 + \beta \in R[x_1]\\
x_2 &\to& F_2 (x_1 , x_2) = \alpha_2 x_2 + f (x_1) \in R[x_1,x_2]\\
x_i &\to& F_i (x_1 \cdots x_i ) = \alpha_i x_i + f_i (x_1 \cdots x_{i-1}) \in R[x_1 \cdots x_i]\\
\end{eqnarray*}
\end{definition}

\begin{definition}
$Af_n (R)$ are those transformations of the form

\begin{eqnarray*}
\vec{x} &\to& A \vec{x} + \vec{b}\\
A &\in& GL_n (R)\\
\end{eqnarray*}

\end{definition}

\begin{theorem}[Jung, van der Kulk \cite{Dicks}]
The group of polynomial automorphisms of $k[x,y]$ denoted $GA_2$is generated as

\begin{eqnarray*}
GA_2 (k) &\simeq& Af_2 (k) \star_{Bf_2(k)} J_2(k)
\end{eqnarray*}

\end{theorem}

\begin{lemma}[\cite{Essen} \label{ZpolAut}]
For an integral domain $D$, the tame subgroup of $Aut D[x,y]$ is the subgroup generated as an amalgamated product $Af_2$ and $J_2$. However this is a proper subgroup, because there exists non-tame Nagata automorphisms like the following for all $a \neq 0$ non-unit:

\begin{eqnarray*}
X &\to& X + a(aY - X^2)\\
Y &\to& Y + 2X(aY - X^2) + a (aY - X^2)^2\\
\end{eqnarray*}

As such we get more complications on $\mathbb{Q}[x_1 \cdots x_n]$ by letting $D_2=\mathbb{Q}[x_1 \cdots x_{n-2}]$. Or on $\mathbb{Z}[x_1 \cdots x_n]$ by letting $D_2 = \mathbb{Z}[x_1 \cdots x_{n-2}]$

\end{lemma}

There are many interesting subgroups inside $GA_{n}(\mathbb{Q})$. In particular inside the $GL_n (\mathbb{Q})$ subgroup, we have:

\begin{theorem}[\label{Feit} \cite{Friedland}]
Finite subgroups of $GL_n \mathbb{Q}$ with maximal order are characterized. Except for $2,4,6,7,8,9,10$ they are the orthogonal groups which have $2^n n!$.

\begin{itemize}
\setlength\itemsep{-1em}
\item $n=2$ has a $W(G_2)$ with order $12 > 8$.\\
\item $n=4$ has a $W(F_4)$ with order $1152 > 384$\\
\item $n=6$ has a $W(E_6) \times \mathbb{Z}_2$ with order $103680 > 46080$\\
\item $n=7$ has a $W(E_7)$ with order $2,903,040 > 645,120$\\
\item $n=8$ has a $W(E_8)$ with order $696,729,600 > 10,321,920$\\
\item $n=9$ has a $W(E_8) \times W(A_1)$ with order $1,393,459,200 > 185,794,560$\\
\item $n=10$ has a $W(E_8) \times W(G_2)$ with order $8,360,755,200 > 3,715,891,200$\\
\end{itemize}

When one further demands matrices over the natural numbers with natural number inverse, there are simply the permutation matrices.

\end{theorem}

\subsection{Non-invertible (bi)modules}

To characterize not necessarily invertible (bi)modules over cluster algebras, we consider coherent sheaves on either $Spec R$ and $Spec R \times_{Spec \mathbb{Q}} Spec R$. However we are only looking at specific (bi)modules, not accounting for maps between (bi)modules. In the discretized 1+1 picture of the Bethe ansatz, (bi)module morphisms would happen at codimension 2 at points on the walls. That means we should ignore the morphisms in these categories and only look at isomorphism classes of objects or even $K_0$ (split or ordinary).

\begin{example}
Consider a polynomial ring $\mathbb{Q}[x_1 \cdots x_n]$. If we restrict our attention to finitely generated projectives, we get that they are all free by Quillen-Suslin \cite{Lam}. This is similarly useful for bimodules by taking $\mathbb{Q}[x_1 \cdots x_n , y_1 \cdots y_n]$. That is $K_0 (\mathbb{Q}[x_1, \cdots x_n]) = \mathbb{Z}$
\end{example}

\section{$Pic$ and $K_0$ of these cluster algebras}

Now let us put all the pieces together to caluclate the associated $Pic$ and $Pic_{com}$ for both $K_0 \mathcal{C}_l \otimes \mathbb{Q}$ and it's exchangeable parts. We can get some information about the noninvertible parts via some invariants related to $K_0$.

\begin{example}[$A_1$]
This has 2 clusters $x$ and $w/x$ and the full cluster algebra is Laurent polynomials $\mathbb{Q}[x^\pm]$. This has a $\mathbb{Z}_2$ for automorphisms over $\mathbb{Q}$ and a trivial $Pic_A (A)$. This means we have a $\mathbb{Z}_2$ invariant. This is the cluster algebra that shows up as the principal part for $Q_{A_1,1}$. $\mathcal{A}$ itself is a polynomial ring in $n*(l+1)=2$ generators. It's automorphisms and Picard group are covered by \cref{polAut} and \cref{polPic}. That means our coarse invariant for invertible $\mathcal{C}_{A_1,1}$ bimodule categories is valued in $Aut_{\mathbb{Q}-alg} (\mathcal{A}) \ltimes Pic_{com} (\mathcal{A}) \simeq GA_2 (\mathbb{Q})$

\begin{tikzcd}
Pic_{com} ( \mathcal{A}) = 0 \arrow[d] & & Aut_{cl}( \mathcal{A} ) \arrow[d,hookrightarrow] \arrow[r,hookrightarrow] & Aut_{\mathbb{Q}-alg} (\mathcal{A}) = GA_2 (\mathbb{Q}) \\
Pic_{com} ( \mathcal{A}^{ex}) = 0 & & Aut_{cl}( \mathcal{A}^{ex} ) = \mathbb{Z}_2 \arrow[r,"\simeq"] & Aut_{\mathbb{Q}-alg} (\mathcal{A}^{ex}) \\
\end{tikzcd}

\end{example}

\begin{example}[$A_2$]
This example is more interesting . It has $(x_1,x_2)$ and 4 other clusters. Altogether $\mathcal{A} \subset \mathbb{Q}[x_1^\pm , x_2^\pm] \subset \mathbb{Q}(x_1,x_2)$. It is $\mathbb{Q}[x_1 , \frac{1+x_1}{x_2} , \frac{1+x_2}{x_1}] \simeq \mathbb{Q}[u,v,w]/(uvw-u-v-1)$ for which $\text{Spec} \mathcal{A}$ which cuts a 2 dimensional hypersurface in affine 3 space. Closure in projective space as $uvw - uz^2 - vz^2 - z^3$ which is a cubic in $P^3$.\\
As above it's cluster automorphisms contain $D_5$, but it's Picard group involves some more algebraic geometry of this cubic.\\
This quiver shows up as the principal part of $\mathcal{A}_{A_1,2}$ and $\mathcal{A}_{A_2,1}$. The $\mathcal{A}$ for these are polynomial rings in $3$ and $4$ generators respectively. Our coarse invariant for bimodule categories for these is valued in $Aut_{\mathbb{Q}-alg} (\mathcal{A}) \ltimes Pic_{com} (\mathcal{A})$ where

\begin{tikzcd}
Pic_{com} ( \mathcal{A}) = 0 \arrow[d] & & Aut_{cl}( \mathcal{A} ) \arrow[d,hookrightarrow] \arrow[r,hookrightarrow] & Aut_{\mathbb{Q}-alg} (\mathcal{A}) = GA_{3/4} (\mathbb{Q}) \\
Pic_{com} ( \frac{\mathbb{Q}[u,v,w]}{(uvw-u-v-1)} ) & & Aut_{cl}( \mathcal{A}^{ex} ) = D_5 \arrow[r,hookrightarrow] & Aut_{\mathbb{Q}-alg} (\mathcal{A}^{ex}) \\
\end{tikzcd}

The difference for both cases is indicated with the $3/4$ and $Aut_{cl} (\mathcal{A})$ may differ between cases.

\end{example}

\begin{example}[$A_3$].\\
The next $A_3$ has $14$ clusters that altogether form $\mathbb{Q}[x_1 , x_3 , \frac{1+x_2}{x_1}=w , \frac{1+x_2+x_1 x_3}{x_2 x_3}=t]$. This can be written as $\mathbb{Q}[x_1 , x_3 , w, t]$ quotiented by a single relation $t w x_1 x_3 - t x_3 - w x_1 - x_1 x_3 = 0$. Upon projectivization of the associated affine scheme, it becomes a quartic hypersurface in $P^4$. This gives the cluster algebra $\mathcal{A}^{ex}$.\\
.\\It's cluster automorphisms contain $D_6$ and computing the Picard group requires more algebraic geometry.\\
This quiver shows up as the principal part of $\mathcal{A}_{A_1,3}$ and $\mathcal{A}_{A_3,1}$. With the frozen part, polynomial rings in $4$ or $6$ variables. Our coarse invariant for bimodule categories for these is valued in $Aut_{\mathbb{Q}-alg} (\mathcal{A}) \ltimes Pic_{com} (\mathcal{A})$ where

\begin{tikzcd}
Pic_{com} ( \mathcal{A}) = 0 \arrow[d] & & Aut_{cl}( \mathcal{A} ) \arrow[d,hookrightarrow] \arrow[r,hookrightarrow] & Aut_{\mathbb{Q}-alg} (\mathcal{A}) = GA_{4/6} (\mathbb{Q}) \\
Pic_{com} ( \frac{\mathbb{Q}[x_1 , x_3 , w, t]}{(t w x_1 x_3 - t x_3 - w x_1 - x_1 x_3)} ) & & Aut_{cl}( \mathcal{A}^{ex} ) = D_6 \arrow[r,hookrightarrow] & Aut_{\mathbb{Q}-alg} (\mathcal{A}^{ex}) \\
\end{tikzcd}

The differences are indicated as above.

\end{example}

\begin{proof}
\begin{eqnarray*}
w x_1 - 1 &=& x_2\\
t x_2 x_3 &=& t (w*x_1 - 1) x_3 = 1 + x_2 + x_1 x_3\\
&=& w x_1 + x_1  x_3\\
t (w x_1 - 1)x_3 &=& w x_1 + x_1  x_3\\
t w x_1 x_3 - t x_3 - w x_1 - x_1 x_3 &=& 0\\
t w x_1 x_3 - t x_3 z^2 - w x_1 z^2 - x_1 x_3 z^2 &=& 0\\
\end{eqnarray*}

\end{proof}

More detailed proofs for the above examples as cluster algebras can be found in \cite{Lampe}.

\begin{theorem}
In general for $\mathcal{C}_{\mathfrak{g} ,l} \subset Rep U_q \hat{\mathfrak{g}}$ in the cases where Leclerc's conjecture is proven, we have an invariant for invertible bimodules valued in $GA_{n(l+1)}$. There is no information about invertible modules contained in this procedure. This structure also fits in diagrams of the form which gets access to more manageable parts of $GA_{n(l+1)}$

\begin{tikzcd}
Pic_{com} ( \mathcal{A}) = 0 \arrow[d] & & Aut_{cl}( \mathcal{A} ) \arrow[d,hookrightarrow] \arrow[r,hookrightarrow] & Aut_{\mathbb{Q}-alg} (\mathcal{A}) = GA_{n(l+1)} (\mathbb{Q}) \\
Pic_{com} ( \mathcal{X}_{f.t.} ) & & Aut_{cl}( \mathcal{A}^{ex} ) = G \arrow[r,hookrightarrow] & Aut_{\mathbb{Q}-alg} (\mathcal{A}^{ex}) \\
\end{tikzcd}

where $\mathcal{X}_{f.t}$ is the specified cluster variety of finite type for the appropriate Dynkin diagram and $G$ is listed in \cref{clustAut}\\

When asking about not necessarily invertible finitely generated projective (bi)modules.

\begin{tikzcd}
K_0 ( \mathcal{A}) = \mathbb{Z} \arrow[r] & K_0 ( \mathcal{A}^{ex})\\
K_0 ( \mathcal{A} \otimes_\mathbb{Q} \mathcal{A}) = \mathbb{Z} \arrow[r] & K_0 ( \mathcal{A}^{ex} \otimes_\mathbb{Q} \mathcal{A}^{ex})
\end{tikzcd}

\end{theorem}

\begin{cor}
For a ``spin" system with $\mathfrak{g}=A_1$ at arbitrary $l \geq 2$, we get a $GA_{l+1} (\mathbb{Q})$ group for $Pic(\mathcal{A}_{A_1,l})$. The invertible bimodules for $\mathcal{A}^{ex}$ have $D_{l+3} = Aut_{cl} (\mathcal{A}^{ex}) \subset Aut (\mathcal{A}^{ex})$ giving $Pic(\mathcal{A}^{ex}) \simeq Aut (\mathcal{A}^{ex}) \ltimes Pic_{com} (\mathcal{A}^{ex})$.
The noninvertible finitely generated projective (bi)modules have a $\mathbb{Z}$ characterization when looking at $\mathcal{A}$. 
\end{cor}

\begin{remark}
If we had a bimodule category fusion of defects, then the cyclic subgroup determined by our favorite defect would have a map to $GA_{l+1}$. In particular if it was finite, we could have a hope of landing in a finite subgroup of maximal order. When $l=3$ or $5 \leq l \leq 9$ \cref{Feit} shows that they would be especially interesting. This might address questions of orders of defects in the q-fermion.
\end{remark}

\begin{cor}
For $\mathcal{A}_{A_{n},1}$ we can also be more specific because the exchangeable part also forms an $A_n$ finite quiver. That is the same exchangeable part as $\mathcal{A}_{A_1,n}$ above. $GA_{2n}$ replaces $GA_{l+1}$ and $Aut_{cl} (\mathcal{A})$ is different because of the different frozen variables. Everything else we have described is independent of the number of variables in the polynomial ring or only involves $\mathcal{A}^{ex}$.
\end{cor}

\begin{lemma}
Let $\mathcal{A}^{ex}$ be the cluster algebra given as the homogenous coordinate ring of a Grassmannian. Then $Cl(\mathcal{A}^{ex})=0$.
\end{lemma}

\begin{proof}
The following correspond to the finite type cluster algebras \cite{Scott} 

\begin{itemize}
\setlength\itemsep{-1em}
\item $Gr(2,n+3)$ for $A_n$\\
\item $Gr(3,6)$ for $D_4$\\
\item $Gr(3,7)$ for $E_6$\\
\item $Gr(3,8)$ for $E_8$\\
\end{itemize}

But these are unique factorization domains \cite{Laface}. This then implies the Weil class group $Cl(Spec \mathcal{A}^{ex})$ is trivial. We also know that these examples are locally acyclic, therefore there can be at worst canonical singularities \cite{BenitoMuller}.

For $K_0$ we can apply the results of \cite{CHWW} for $\mathcal{A}^{ex}$ the homogenous coordinate ring of $X$ of dimension $d$ to give

\begin{eqnarray*}
K_0 ( \mathcal{A}^{ex}) &=& \mathbb{Z} \bigoplus Pic (\mathcal{A}^{ex}) \bigoplus_{p=1}^{d} \bigoplus_{k=1}^{\infty} H^p (X ,\Omega_X^p (k))\\
\end{eqnarray*}

Then \cite{Snow} gives results for each of the $H^p (X ,\Omega_X^p (k))$

\end{proof}

\iftrue
\section{Semiclassical Geometry}

Cluster algebras often also come as Poisson algebras. This means that we should classify (invertible) (bi)modules for these deformations. This deformation may or may not come from a deformation of the $\mathcal{C}_{\mathfrak{g},l} \subset Rep \; U_q \hat{\mathfrak{g}}$. Candidates for this may be possible using elliptic quantum groups \cite{Felder,ToledanoLaredo}.

\begin{definition}[Classical Limit]
Let $\mathcal{A}$ be the starting commutative unital Poisson algebra over the commutative unital ring $k$. (For us $\mathbb{Q}$). 
Then define $\mathcal{Q}$ to be $\mathcal{A}[[\hbar]]$ \footnote{$\hbar$ may not be the best notation because the system was already quantum. This is yet another deformation.} with product $\star$ as well as the classical limit maps $cl$

\begin{eqnarray*}
\mathcal{Q} &\to& \mathcal{A}\\
\sum_{r=0}^\infty \hbar a_r &\to& a_0\\
Pic (\mathcal{Q}) &\to& Pic( \mathcal{A})\\
K_0 ( \mathcal{Q}) &\to& K_0 ( \mathcal{A})\\
\end{eqnarray*}
\end{definition}

\begin{theorem}[\cite{Rosenberg96}]
The map on $K_0$ is an isomorphism of groups. The not necessarily invertible have no changes.
\end{theorem}

\begin{theorem}[\cite{Bursztyn}]
The kernel of $cl$ on $Pic$ is in 1-1 correspondence with outer self-equivalences of $\mathcal{Q}$. The image can be described in terms of the action of $Pic(\mathcal{A})$ on deformations through gauge transformation equivalence. This is also called a B-field transform.
\end{theorem}

In the case where taking $\mathbb{R}$ points has given a real Poisson manifold $M$, one can ask the differential geometric analog\footnote{See \cite{Weinstein} for the Morita theory for Poisson manifolds.} and with complex valued functions here we get:

\begin{theorem}[7.1 of \cite{BursztynWaldmann}]

Let the deformation $\mathcal{Q}$ with $\star$ be such that:

\begin{itemize}
\setlength\itemsep{-1em}
\item There exists a linear map from the Poisson center $\mathcal{Z}_\pi (\mathcal{A})$ to the center $Z(\mathcal{Q})$ with $f \to f+ O(\hbar)$\\
\item There exists a linear map $PDer(\mathcal{A}) \to Der(\mathcal{Q})$ where the vector field $X \to \mathcal{L}_X + O(\hbar)$ and for Hamiltonian vector fields, $X_H \to \frac{i}{\hbar} ad_H$
\end{itemize}

then for the classical limit map on $Pic$. 

\begin{eqnarray*}
ker \; cl &\leftrightarrow& \frac{H_\pi^1 (M ,\mathbb{C})}{2\pi i H_\pi^1 ( M , \mathbb{Z})} + \hbar H_\pi^1 (M,\mathbb{C})[[\hbar]]\\
H_\pi^1 (M , \mathbb{C}) =  0 &\implies& ker \; cl = 0\\ &\implies& Pic(\mathcal{Q}) \hookrightarrow Diff(M) \ltimes H^2 (M,\mathbb{Z})
\end{eqnarray*}

\end{theorem}

\begin{cor}
If in addition $\pi = \omega^{-1}$ and $H^1 ( M ) = 0$, then $Pic( \mathcal{Q}) \hookrightarrow Diff(M) \ltimes H^2 (M , \mathbb{Z})$
\end{cor}

\begin{proof}
In the symplectic case $H^1 (M) \simeq H_\pi^1 (M)$.
\end{proof}

\fi

\section{Conclusion}

Taking inspiration from the algebraic Bethe ansatz with boundaries and defects we have considered the characterization of (invertible) (bi)modules for the associated cluster algebras. We have found interesting groups that might or might not lift to the categorical level of the spin chain.

There are many questions this raises. Among these are producing actual bimodule categories for these $\mathcal{C}_{\mathfrak{g},l}$. Our main tool for producing interesting bimodule categories is covariantization of coquasitriangular algebras \cite{self,KolbStokman}. We would also like to work in the RTT realization. The difference of coproducts means physical properties will be drastically different between these cases.

Another question is if the $\otimes \mathbb{Q}$ can be avoided. This is because the algebras were actually defined over $\mathbb{Z}$. Even for polynomial rings we get complications that can be computed as \cref{ZpolPic} and \cref{ZpolAut}. That gives the interesting group $GA_{n(l+1)} (\mathbb{Z})$. The exchangeable part which is not a polynomial algebra will have even more complications. For example, the singularity for $A_n$ when $n \equiv 3 (\text{mod} 4)$ \cite{Muller}. The Picard groups for these can be tackled with \cref{Weibel}. Life in a ring is harder than life in a field. An infinitesimal formal degree $2$ parameter can be used to tropicalize by defining $\xi_i = T \log x_i$. That resembles the case of equivariant K-theory of Grassmannians \cite{Smirnov}\footnote{The homological degee provides a caution for when one can expect convergence for numerical nonzero values of a parameter.}. $K_0 ( \mathcal{C}_{A_n,l} ) \otimes_{\mathbb{Z}} \mathbb{C}$ is also a quotient of the homogenous  coordinate ring of $Gr(n+1,n+l+2)$ so relations with the Amplitudehedron are possible \cite{Nima}.

In all cases we have provided a map $\mathbb{Z} \to K_0 (\mathcal{A}^{ex})$. We would like to calculate the image of $1$ and upon rationalization get these special elements of the Chow rings for all $\mathfrak{g}$ and $l$. The same is true for the bimodule case. We also described the analogous case for formal deformations. The differential geometric $Pic$ changed in a reasonably manageable way, but we do not know about the algebraic side.

\bibliographystyle{ieeetr}
\bibliography{clusterCategoryDefect2}

\end{document}